\documentclass[11pt,a4paper]{article}
\usepackage{fullpage}
\usepackage{microtype}
\usepackage{mathtools,amsmath} \usepackage{amssymb,amsthm,lineno}
\usepackage{graphicx} \usepackage{xspace} \usepackage{xypic}
\usepackage[numbers]{natbib} \usepackage{hyperref}
\usepackage{color,pdfpages}
\graphicspath{{figures/}}
\theoremstyle{plain} 

\newtheorem{theorem}{Theorem} 
\newtheorem{lemma}[theorem]{Lemma}

\newtheorem{corollary}[theorem]{Corollary}

\title{Point Location in Dynamic Planar Subdivisions%
	\footnote{This research was supported by NRF grant 2011-0030044 (SRC-GAIA) funded by the government of Korea,
		and the MSIT (Ministry of Science and ICT), Korea, under the SW Starlab support program (IITP-2017-0-00905) 
		supervised by the IITP (Institute for Information \& communications Technology Promotion).}}

\author{Eunjin Oh\thanks{Max Planck Institute for Informatics, {Saarbr\"{u}cken, Germany, Email: {\tt{eoh@mpi-inf.mpg.de}}}} \and
	Hee-Kap Ahn\thanks{Pohang University of Science and Technology, {Pohang, Korea}, Email: {\tt{heekap@postech.ac.kr}}}}

\newcommand{\InsertEdge}{\textsc{InsertEdge}}
\newcommand{\DeleteEdge}{\textsc{DeleteEdge}}
\newcommand{\locate}{\textsc{locate}}

\newcommand{\oldd}{{\mathsf{D_o}}}
\newcommand{\newd}{\mathsf{D_n}}
\newcommand{\oldm}{{\mathsf{M_o}}}
\newcommand{\newm}{\mathsf{M_n}}

\newcommand{\compm}{\mathsf{M_c}}
\newcommand{\mm}{\mathsf{M}}
\newcommand{\ee}{\mathsf{E}}
\newcommand{\newce}{\mathsf{E_n}}
\newcommand{\oldce}{{\mathsf{E_o}}}

\newcommand{\oldf}{{F_\mathsf{o}}}
\newcommand{\newf}{F_\mathsf{n}}

\newcommand{\newe}{e_\mathsf{n}}
\newcommand{\olde}{e_\mathsf{o}}

\newcommand{\setS}{\mathsf{S}}

\graphicspath{{figures/}}

\begin{document}

\maketitle
\begin{abstract}
	We study the point location problem on dynamic planar  
	subdivisions that allows insertions and deletions of edges. In our problem,
		the underlying graph of a subdivision is not necessarily connected.
	We present a data structure of linear size for such a dynamic planar subdivision that supports
	sublinear-time update and polylogarithmic-time query. Precisely, the amortized update
	time is $O(\sqrt{n}\log n(\log\log n)^{3/2})$ and the query time is
	$O(\log n(\log\log n)^2)$, where $n$ is the number
	of edges in the subdivision. This answers a question posed by Snoeyink
	in the Handbook of Computational Geometry.
	When only deletions of edges are allowed, the update time and query time
	are just $O(\alpha(n))$ and $O(\log n)$, respectively.
\end{abstract}

\section{Introduction}

Given a planar subdivision, a point location query asks with a query
point specified by its coordinates to find the face 
of the subdivision containing the query point. In many situations such
point location queries are made frequently, and therefore it is
desirable to preprocess the subdivision and to store it in a data
structure that supports point location queries fast.

The planar subdivisions for point location queries are usually induced
by planar embeddings of graphs. A planar subdivision is connected if
the underlying graph is connected. The vertices and edges of the
subdivision are the embeddings of the nodes and arcs of the underlying
graph, respectively. An edge of the subdivision is considered to be
open, that is, it does not include its endpoints (vertices). A face of
the subdivision is a maximal connected subset of the plane that does
not contain any point on an edge or a vertex.

We say a planar subdivision \emph{dynamic} if the subdivision allows
two types of operations, the insertion of an edge to the subdivision
and the deletion of an edge from the subdivision.  The subdivision
changes over insertions and deletions of edges accordingly.  For an
insertion of an edge $e$, we require $e$ to intersect no edge or
vertex in the subdivision and the endpoints of $e$ to lie on no edge
in the subdivision.  We insert the endpoints of $e$ in the subdivision
as vertices if they were not vertices of the subdivision. In fact, the
insertion with this restriction is general enough.  The insertion of
an edge $e$ with an endpoint $u$ lying on an edge $e'$ of the
subdivision can be done by a sequence of four operations: deletion of
$e'$, insertion of $e$, and insertions of two subedges of $e'$
partitioned by $u$.

The dynamic point location problem is closely related to the dynamic
vertical ray shooting problem~\cite{Cheng-NewResults-1992}.  For this 
problem, we are asked to find the edge of a
dynamic planar subdivision that lies immediately above (or below) a query point.  In
the case that the subdivision is connected at any time, we can
answer a point location query without increasing the space and time
complexities using a data structure for the dynamic vertical ray shooting
problem by maintaining the list of the edges incident to each face in
a concatenable queue~\cite{Cheng-NewResults-1992}.

However, it is not the case in a general (possibly disconnected) planar
subdivision.  Although the dynamic vertical ray shooting algorithms
presented
in~\cite{ABG-Improved-2006,BJM-Dynamic-1994,cn-locatoin-2015,Cheng-NewResults-1992}
work for general (possibly disconnected) subdivisions, it is unclear how one can
use them to support point location queries efficiently.  As pointed
out in some previous
works~\cite{cn-locatoin-2015,Cheng-NewResults-1992}, a main issue
concerns how to test whether the two edges lying immediately above two
query points belong to the boundary of the same face in a
dynamic planar subdivision.  Notice that the boundary of a face may
consist of more than one connected component.

In this paper, we consider a point location query 
on dynamic planar subdivisions. The subdivisions we consider
are not necessarily connected, that is, the underlying graphs may
consist of one or more connected components.  We also require that
every edge is a straight line segment.  We present a data structure for
a dynamic planar subdivision which answers point location queries efficiently.

\subparagraph{Previous work.}
The dynamic vertical ray shooting problem has been studied
extensively~\cite{ABG-Improved-2006,BJM-Dynamic-1994,cn-locatoin-2015,Cheng-NewResults-1992}.
These data structures do not require that the subdivision is connected, but
they require that the subdivision is planar.  None of the known algorithms for
this problem is superior to the others.  Moreover, optimal update and
query times (or their optimal trade-offs) are not known.  The
update time or the query time (or both) is worse than
$O(\log^2 n)$, except the data structures by Arge et
al.~\cite{ABG-Improved-2006} and by Chan and
Nekrich~\cite{cn-locatoin-2015}.  The data structure by Arge at
al.~\cite{ABG-Improved-2006} supports expected $O(\log n)$ query time
and expected $O(\log^2 n/\log\log n)$ update time under Las Vegas
randomization in the RAM model.  The data structure by Chan and
Nekrich~\cite{cn-locatoin-2015} supports $O(\log n(\log \log n)^2)$
query time and $O(\log n \log \log n)$ update time in the pointer
machine model.  Their algorithm can also be modified to reduce the
query time at the expense of increasing the update time.  As pointed
out by Cheng and Janardan~\cite{Cheng-NewResults-1992}, all these data
structures~\cite{ABG-Improved-2006,BJM-Dynamic-1994,cn-locatoin-2015,Cheng-NewResults-1992}
can be used for answering point location queries if the underlying
graph of the subdivision is connected without increasing any resource.


Little has been known for the dynamic point location in general
planar subdivisions.  In fact, no nontrivial data structure
is known for this problem.\footnote{The paper~\cite{ABG-Improved-2006}
		claims that their data structure supports a point location query for
		a general subdivision. They present a vertical ray shooting data
		structure and claim that this structure supports a point location
		query for a general subdivision using the paper~\cite{overmars}.
		However, the paper~\cite{overmars} mentions that it works only for a
		subdivision such that every face in the subdivision has a constant
		complexity. Therefore, the point location problem for a general
		subdivision is still open.} Cheng and Janardan asked whether such a data structure can
be maintained for a general planar
subdivision~\cite{Cheng-NewResults-1992}, but this question has not
been resolved until now.  Very recently, it was asked again by Chan
and Nekrich~\cite{cn-locatoin-2015} and by
Snoeyink~\cite{handbook}.  Specifically, Snoeyink asked
whether it is possible to construct a dynamic data structure for a
general (possibly disconnected) planar subdivision supporting
\emph{sublinear query time} of determining if two query points
lie in the same face of the subdivision.

\subparagraph{Our result.}
In this paper, we present a data structure and its update and query
algorithms for the dynamic point location in general planar 
subdivisions under the pointer machine model. 
This is the first result supporting sublinear update and
query times, and answers the question posed
in~\cite{cn-locatoin-2015,Cheng-NewResults-1992,handbook}.  Precisely,
the amortized update time is $O(\sqrt{n}\log n(\log\log n)^{3/2})$ and the query time is
$O(\log n(\log\log n)^2)$, where $n$ is the number of edges in the
current subdivision.  
When only deletions of edges are allowed, the update and query
times are just $O(\alpha(n))$ and $O(\log n)$, respectively. Here, we assume that a
deletion operation is given with the \emph{pointer} to an
edge to be deleted in the current edge set.

Our approach itself does not require that every edge in the
subdivision is a line segment, and can handle arbitrary
curves of constant description.
However, the data structures for
dynamic vertical ray shooting queries require that every edge is a straight
line segment, which we use as a black box.  Once we have a data
structure for answering vertical ray shooting queries for general
curves, we can also extend our results to general curves.  For
instance, the result by Chan and Nekrich~\cite{cn-locatoin-2015} is
directly extended to $x$-monotone curves, and so is ours.

One may wonder if the problem is \emph{decomposable} in the sense that
a query over $D_1\cup D_2$ can be answered in constant time from the
answers from $D_1$ and $D_2$ for any pair of disjoint data sets $D_1$
and $D_2$~\cite{matousek-ept-1992}.  If a problem is decomposable, we
can obtain a dynamic data structure from a static data structure of
this problem using the framework of Bentley and
Saxe~\cite{Bentley1980}, or Overmars and
Leeuwen~\cite{ol-dynamic-1981}.  However, the dynamic point location
problem in a general planar subdivision is not decomposable. To see this,
consider a subdivision $D$ consisting of a square face and one
unbounded face. Let $D_1$ be the subdivision consisting of three edges
of the square face and $D_2$ be the subdivision consisting of the
remaining edge of the square face.  There is only one face in $D_1$
(and $D_2$). Any two points in the plane are contained in the same
face in $D_1$ (and $D_2$). But it is not the case for $D$. Therefore,
the answers from $D_1$ and $D_2$ do not help to answer point location
queries on $D$.

\subparagraph{Outline.}
Consider any two query points in the plane. Our goal is to check
whether they are in the same face of the current subdivision.  To do
this, we use the data structures for answering dynamic vertical ray shooting
queries~\cite{ABG-Improved-2006,BJM-Dynamic-1994,cn-locatoin-2015,Cheng-NewResults-1992},
and find the edges lying immediately above the two query points.  Then
we are to check whether the two edges are on the boundary of the same
face.  In general subdivisions, the boundary of each face may consist
of more than one connected components. This makes $\Theta(n)$
changes to the boundaries of the faces in the worst case, where $n$
is the number of edges in the current subdivision.  Therefore, we
cannot maintain the explicit description of the subdivision. 

To resolve this problem, we consider two different subdivisions,
$\oldm$ and $\newm$, such that 
 the current subdivision consists of the edges of $\oldm$ and $\newm$,
and construct data structures on the subdivisions,
$\oldd$ and $\newd$, respectively.  Recall that the dynamic point
location problem is not decomposable. Thus the two subdivisions
must be defined carefully.
We set each edge in the current
subdivision to be one of the three states: old, communal, and new.
Then let $\oldm$ be the subdivision induced by all old and communal
edges, and $\newm$ be the subdivision induced by all new and communal
edges.  Note that every communal edge belongs to both subdivisions.

The state of each edge is defined as follows.  The data structures are
reconstructed periodically. In specific, they are rebuilt
after processing $f(n)$ updates since the latest
reconstruction, where $n$ is the number of edges in the current
subdivision. Here, $f(n)$ is called a reconstruction period, which is
set to $\sqrt{n}$ roughly.  When an edge $e$ is inserted, we find the
face $F$ in $\oldm$ intersecting $e$ and set the old edges on the outer
boundary of $F$ to communal. If one endpoint of $e$ lies on the outer 
boundary of $F$ and the other lies on an inner boundary of $F$,
we set the old edges on this inner boundary to communal. Also, we set $e$ to new.  When an
edge $e$ is deleted, we find the faces in $\oldm$ incident to $e$ and
set the old edges of the outer boundaries of the faces to communal.

We show that the current subdivision has the following property: 
no
face in the current subdivision contains both new and old edges on its
outer boundary. In other words, 
	for every face in the current subdivision, either 
	every edge is classified as new or communal, or every edge
	is classified as old or communal. 
Due to this property, for any two query points, they
are in the same face in the current subdivision if and only if they
are in the same face in both $\oldm$ and $\newm$.  Therefore, we can
represent the name of a face in the current subdivision as a pair of
faces, one in $\oldm$ and one in $\newm$.  To answer a point location
query on the current subdivision, it suffices to find the faces
containing the query point in $\oldm$ and in $\newm$.

To answer point location queries on $\oldm$, we observe that no edge
is inserted to $\oldm$ unless it is rebuilt.  Therefore, it suffices
to construct a semi-dynamic point location data structure on $\oldm$
supporting only deletion operations.  If only deletion operations are
allowed, two faces are merged into one face, but no new face
appears. Using this property, we provide a data structure
supporting $O(\alpha(n))$ update time and $O(\log n)$ query time.

To answer point location queries on $\newm$, we make use of the
following property: the boundary of each face of $\newm$ consists of
$O(f(n))$ connected components while
the number of edges of $\newm$ is $\Theta(n)$ in the worst case,
where $n$ is the number of all edges in the current subdivision. Due
to this property, the amount of the change on the subdivision $\newm$
is $O(f(n))$ at any time.  Therefore, we can maintain the explicit
description of $\newm$.  In specific, we maintain a data structure on
$\newm$ supporting point location queries, which is indeed
a doubly connected linked list of $\newm$.

Due to lack of space, some proofs and details are omitted. The missing proofs and
missing details can be found in the full version of the paper.

\section{Preliminaries}
Consider a planar subdivision $\mm$ that consists of $n$ 
straight line segment edges.
Since the subdivision is planar, there are $O(n)$ vertices and faces.
One of the faces of $\mm$ is unbounded and all other faces are bounded.
Notice that the boundary of a face is not necessarily
connected. 
For the definitions of the faces and their boundaries, refer to~\cite[Chapter 2]{CGbook}.


We consider each edge of the subdivision as two directed
\emph{half-edges}.  The two half-edges are oriented in opposite
directions so that the face incident to a half-edge lies to the left
of it. In this way, each half-edge is incident to exactly one face,
and the orientation of each connected component of the boundary of $F$
is defined consistently.  We call a boundary component of $F$ the
\emph{outer boundary} of $F$ if it is traversed along its half-edges
incident to $F$ in counterclockwise order around $F$. Except for the
unbounded face, every face has a unique outer boundary. We call each
connected component other than the outer boundary an \emph{inner
	boundary} of $F$.  Consider the outer boundary $\gamma$ of a
face. Since $\gamma$ is a noncrossing closed curve, it subdivides the plane into regions exactly one of which contains $F$.
We say a face $F$ \emph{encloses} a set $C$ in the plane if $C$ is contained in 
the (open) region containing $F$ of the planar subdivision induced by the outer boundary of $F$.
Note that if $F$ encloses $F'$, the
outer boundary of $F$ does not intersect the boundary of $F'$.  For
more details on planar subdivisions, refer to the computational
geometry book~\cite{CGbook}.

Our results are under the pointer machine model, which is more
restrictive than the random access model.  Under the pointer machine
model, a memory cell can be accessed only through a series of pointers
while any memory cell can be accessed in constant time under the
random access model. Most of the results
in~\cite{ABG-Improved-2006,BJM-Dynamic-1994,cn-locatoin-2015,Cheng-NewResults-1992}
are under the pointer machine model, and the others are under the
random access model.

\subparagraph{Updates: insertion and deletion of edges.}  We allow two
types of update operations: $\InsertEdge(e)$ and $\DeleteEdge(e)$.  In
the course of updates, we maintain a current edge set $\ee$, which is
initially empty.  $\InsertEdge(e)$ is given with an edge $e$ such that
no endpoints of $e$ lies on an edge of the current
subdivision.  This operation adds $e$ to $\ee$, and thus update the
current subdivision accordingly.  Recall that an edge of the
subdivision is a line segment excluding its endpoints. If an endpoint
of $e$ does not lie on a vertex of the current subdivision, we also
add the endpoint of $e$ to the current subdivision as a vertex.
$\DeleteEdge(e)$ is given with an edge $e$ in the current
subdivision. Specifically, it is given with a pointer to $e$ in the
set $\ee$. This operation removes $e$ from $\ee$, and updates the
subdivision accordingly.  If an endpoint of $e$ is not incident to any
other edge of the subdivision, we also remove the vertex which is the
endpoint of $e$ from the subdivision.
%

\subparagraph{Queries.}
Our goal is to process update operations on the data structure
so that given a query point $q$ the face of the current subdivision 
containing $q$ can be computed from the data structure efficiently.
Specifically, each face is assigned a distinct
name in the subdivision, and given a query point
the name of the face containing the point is to be reported.
A query of this type is called a \emph{location query}, denoted
by $\locate(x)$ for a query point $x$ in the plane.

\subsection{Data structures}\label{sec:facetree}
In this paper, we show how to process updates and queries 
efficiently by
maintaining a few data structures for dynamic planar subdivisions.  In
specific, we use disjoint-set data structures and
concatenable queues.
Before we continue with algorithms for updates and queries, we provide
brief descriptions on these structures in the
following. 
Throughout this paper, we use $S(n), U(n)$ and $Q(n)$ to denote the
size, the update and query time of the data structures we use for the
dynamic vertical ray shooting in a general 
subdivision.
Notice that $U(n)=\Omega(\log n)$, $U(n)=o(n)$, and $Q(n)=\Omega(\log n)$
	for any nontirivial data structure  for the
	dynamic vertical ray shooting problem 
	under the pointer machine model. Also, $U(n)$ is increasing. 
	Thus in the following, we assume that $U(n)$ and $Q(n)$ satisfy
	these properties.

A disjoint-set data structure keeps track of
a set of elements partitioned into a number of
disjoint subsets~\cite{Tarjan-1975}. Each subset is represented by a rooted
tree in this data structure. 
The data structure has size linear in the total number of
elements, and can be used to check whether two elements
are in the same partition and to merge two partitions into one.
Both operations can be done in $O(\alpha(N))$ time, where $N$ is the
number of elements at the moment and $\alpha(\cdot)$
is the inverse Ackermann function.

A concatenable queue represents a sequence of elements, and allows four operations: insert an 
element, delete an element, split the sequence into two subsequences,
and concatenate two concatenable queues into
one. By implementing them with 2-3 trees~\cite{Aho}, we can support
each operation in $O(\log N)$ time, where $N$ is the number of
elements at the moment. We can search any element in
the queue in $O(\log N)$ time.

\section{Deletion-only point location}\label{sec:decremental}
In this section, we present a semi-dynamic data structure for point
location queries that allows only $\DeleteEdge$ operations.
Initially, we are given a planar subdivision consisting of $n$ edges. Then we
are given update operations $\DeleteEdge(e)$ for edges $e$ in the
subdivision one by one, and process them accordingly.
In the course of updates, we
answer point location queries.
We maintain static data structures on the initial subdivision and a disjoint-set data structure
that changes dynamically as we process $\DeleteEdge$ operations.

\subparagraph{Static data structures.} 
We construct the static point location data structure on the initial
subdivision of size $O(n)$ in $O(n\log n)$
time~\cite{ST-location-1986}. Due to this data structure, we can find
the face in the initial subdivision containing a query point in
$O(\log n)$ time.  We assign a name to each face, for instance, the
integers from $1$ to $m$ for $m$ faces.
Also, we compute the doubly connected edge list of the initial subdivision, and make each edge in the current edge set to point to its counterparts in
the doubly connected edge list.
These data structures are static, so they do not change in the course of updates.

\subparagraph{History structure of faces over updates.}
Consider an edge $e$ to be deleted from the subdivision. If $e$ is
incident to two distinct faces in the current subdivision, the faces
are merged into one and the subdivision changes accordingly.  To
apply such a change and keep track of the face information of the
subdivision, we use a disjoint-set data structure $\setS$ on the names
of the faces in the initial subdivision. Initially,
each face forms a singleton subset in $\setS$.  In the course of
updates, subsets in $\setS$ are merged.  Two elements in $\setS$ are
in the same subset of $\setS$ if and only if the two faces
corresponding to the two elements are merged into one face.

\subparagraph{Deletion.}  We are given $\DeleteEdge(e)$ for
an edge $e$ of the subdivision.  Since only history structure changes
dynamically, it suffices to update the history structure only.  We
first compute the two faces in the initial subdivision that are
incident to $e$ in $O(1)$ time by using the doubly connected edge
list.  Then we check if the faces belong to the same subset or not in
$\setS$. If they belong to two different subsets, we merge the subsets
into one in $O(\alpha(n))$ time.  The label of the root node in the
merged subset becomes the name of the merged face.  If the faces
belong to the same subset, $e$ is incident to the same face $F$ in the
current subdivision, and therefore there is no change to the faces in
the subdivision, except the removal of $e$ from the boundary of $F$.
Since we do not maintain the boundary information of faces, 
there is nothing to do with the removal and we do not do anything on $\setS$.
Thus, there is a bijection between faces in the current subdivision
and subsets in the disjoint-set data structure $\setS$. We say that
the face corresponding to the root of a subset \emph{represents} the
subset.

\subparagraph{Location queries.}  To answer $\locate(x)$ for a query point $x$
in the plane, we find the face $F$ in the initial subdivision in
$O(\log n)$ time. Then we return the subset in the disjoint-set data
structure $\setS$ that contains $F$ in $O(\alpha(n))$ time. Precisely,
we return the root of the subset containing $F$ whose label is the
name of the face containing $x$ in the current subdivision.
The argument in this section implies the correctness of the query algorithm.

\begin{theorem}\label{thm:decrement}
	Given a planar subdivision consisting of $n$ edges,
	we can construct data structures of size $O(n)$ in $O(n\log n)$ time so
	that $\locate(x)$ can be answered in $O(\log n)$
	time for any point $x$ in the plane and the data
	structures can be updated in $O(\alpha(n))$ time for a
	deletion of an edge from the subdivision. 
\end{theorem}

\section{Data structures for fully dynamic point location}\label{sec:ds}
In Section~\ref{sec:ds} and Section~\ref{sec:update}, we present a
data structure and its corresponding update and query algorithms
for the dynamic point location in fully dynamic planar
subdivisions.  Initially, the subdivision is the whole plane.
While we process a mixed sequence of insertions and deletions of
edges, we maintain two data structures, one containing \emph{old and
	communal edges} and one containing \emph{new and communal edges}.
We consider each edge of the subdivision to have one of three states,
``new'', ``communal'', and ``old''. The first data structure, denoted
by $\oldd$, is the point location data structure on old and communal
edges that supports only $\DeleteEdge$ operations described in
Section~\ref{sec:decremental}. The second data structure, denoted by
$\newd$, is a fully dynamic point location data structure on new and
communal edges.

\begin{figure}
	\begin{center}
		\includegraphics[width=0.7\textwidth]{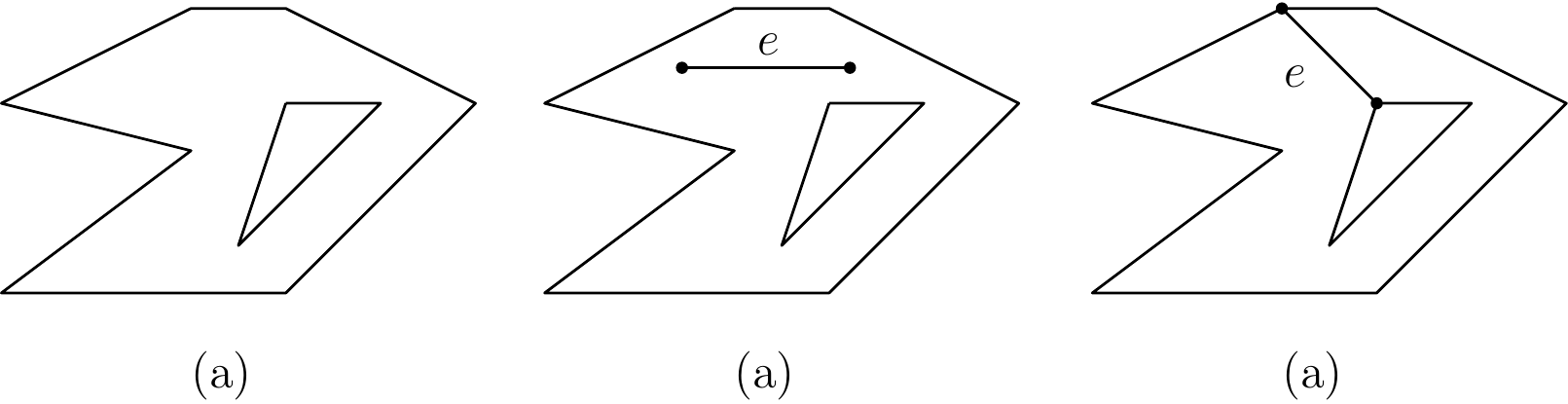}
		\caption{\small(a) Face containing only one inner boundary.
			(b) The insertion of $e$ makes the old edges on the outer boundary communal.
			(c) The insertion of $e$ makes all old boundary edges communal. \label{fig:com}}
	\end{center}
\end{figure}

\subparagraph{Three states: old, communal and new.}
We rebuild both data structures
periodically. When they are rebuilt, every edge in the current
subdivision is set to old. As we process updates, some of them
are set to communal as follows. 
For $\InsertEdge(e)$,
there is exactly one face $F$ of $\oldm$ whose interior is intersected by $e$
as $e$ does not intersect any edges or vertices. 
We set all edges on the outer boundary of $F$ to communal.
If $e$ connects the outer boundary of $F$ with one inner boundary of $F$,
we set all edges on the inner boundary to communal.
As a result, the outer boundary edges in the faces incident to $e$ in $\oldm$
are communal or new after $e$ is inserted. See
Figure~\ref{fig:com}.
For $\DeleteEdge(e)$,
there are at most two faces of $\oldm$ whose boundaries contain $e$. 
We set all edges on the outer boundary of these faces to communal.
Here, we do not maintain the explicit description (the doubly connected edge list) of $\oldm$,
but maintain the semi-dynamic data structure on $\oldm$ described in Section~\ref{sec:decremental}. 

Also, the edges inserted after the latest reconstruction are set to new.
The subdivision $\newm$ of the new and communal edges has complexity of $\Theta(n)$ in the 
worst case. We maintain a fully dynamic point location data structure on $\newm$,
which is indeed the explicit description (the doubly connected edge list) of $\newm$.

\subsection{Reconstruction}
Let $f: \mathbb{N}\rightarrow \mathbb{N}$ be an increasing function
satisfying that $f(n)/2\leq f(n/2)\leq n/4$ for every $n$ larger than
a constant, which will be specified later. We call the function a
\emph{reconstruction period}.  We reconstruct $\oldd$ and $\newd$ if
we have processed $f(n)$ updates since the latest reconstruction time,
where $n$ is the number of the edges in the current subdivision.

The following lemma is a key to achieve an efficient update time.
Each face of $\newm$ has $O(f(n))$ boundary components 
while the number of edge in $\newm$ 
is $\Theta(n)$ in the worst case.

\begin{lemma}\label{lem:new-face}
	Each face of $\newm$ has $O(f(n))$ inner boundaries.
\end{lemma}
\begin{proof} 
	Consider a face $F$ of $\newm$.
	There are two types of the inner boundaries of $F$: either all edges on 
	an inner boundary are communal or at least one edge on an inner boundary is new. For an
	inner boundary of the second type, all edges other than the new edges are communal.  By the construction,
	the number of new edges in $\newm$ is at most $f(n)$.  Recall that
	when we rebuild the data structures, $\oldm$ and $\newm$, all edges
	are set to old.
	
	We pick an arbitrary edge on each inner boundary of $F$
	of the first type, and call it the \emph{representative} of the 
	inner boundary.
	Each representative 
	is inserted before the latest reconstruction,
	and is set to communal later.
	It becomes communal due to a pair $(F',e')$, where $F'$ is a face of
	$\oldm$ and $e'$ is an edge inserted or deleted after the latest
	reconstruction, such that the insertion or deletion of $e'$ makes
	the edges on a boundary component of $F'$ communal, and $e$ was on this boundary component of $F'$
	at that moment.  See
	Figure~\ref{fig:number-face}. 
	If the representatives of all first-type inner boundaries of $F$ are induced by distinct pairs,
	it is clear that the number of the first-type inner boundaries of $F$ is $O(f(n))$.
	But it is possible that the representative of some inner boundaries of $F$ are induced
	by the same pair $(F',e')$. 
	
	Consider the representatives of some inner boundaries of $F$ that are induced by the same pair $(F',e')$.
	The insertion of deletion of $e'$ makes the outer boundary 
	of $F'$ become communal. In the case that $e'$ connects the outer boundary of $F'$ and an
	inner boundary of $F'$, let $\gamma$ be the cycle consisting of the outer boundary of $F'$, 
	the inner boundary of $F'$ and $e'$. Let $\gamma$ be the outer boundary of $F'$, otherwise. 
	All representatives induced by $(F',e')$ are on $\gamma$,
	and therefore they are connected after $e'$ is inserted or before $e'$ is deleted.
	Notice that any two of such representatives are disconnected later.
	This means that $\gamma$ becomes at least $t$ connected components due to the removal of $t$ edges on it,
	where $t$ is the number of the representatives induced by $(F',e')$. 
	The total number of edges that are removed after the latest reconstruction is $O(f(n))$,
	and each edge that are removed after the latest reconstruction can
	be a representative of at most two first-type inner boundaries of $F$. Therefore, the total number of the representatives of the first-type inner boundaries of $F$ is
	also $O(f(n))$.
	
	Consider an inner boundary of the second type. Since each edge is incident
	to at most one inner boundary of $F$, the number of the second-type inner boundaries
	is at most the number of new edges. Therefore,
	there are $O(f(n))$ second-type inner boundaries of $\newm$.
\end{proof}

	\begin{figure}
	\begin{center}
		\includegraphics[width=0.55\textwidth]{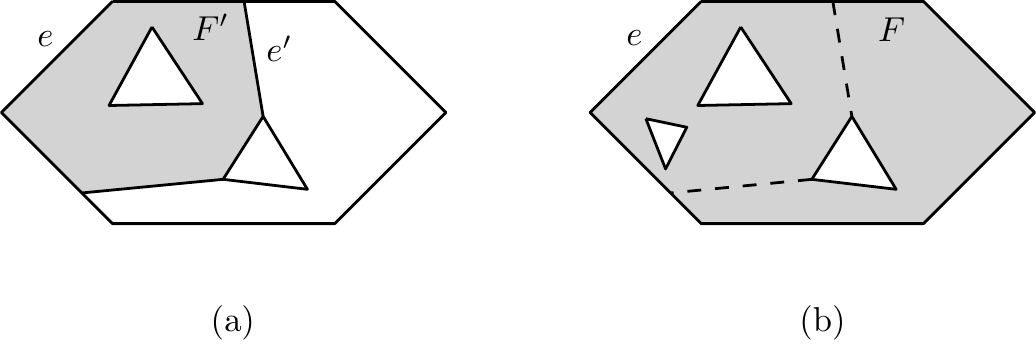}
		\caption{\small(a) Subdivision $\oldm$. All edges are old. (b)
			Subdivision $\newm$. The two dashed edges are deleted after
			the reconstruction, which makes all outer boundary edges
			become communal. Then the edges on the leftmost triangle
			(hole) are inserted. \label{fig:number-face}}
	\end{center}
\end{figure}

\subsection{Two data structures}
We maintain two data structures: $\oldd$, a semi-dynamic point location for old
and communal edges, and $\newd$, a fully dynamic point location for new and communal edges.
In this subsection, we describe the data structures $\oldd$ and $\newd$.
The update procedures are described in Section~\ref{sec:update}.

\subparagraph{Semi-dynamic point location for old and communal edges.}
After each reconstruction, we
construct the point location data structure $\oldd$ supporting only
$\DeleteEdge$ described in Section~\ref{sec:decremental} for all edges
in the current subdivision, which takes $O(n\log n)$ time. Recall that
all edges in the current subdivision are old at this moment.  In
Section~\ref{sec:update}, we will see that the amortized time for
reconstructing $\oldd$ is $O(n\log n/f(n))$ at any moment, where $n$
is the number of all edges in the subdivision at the moment.
As update operations are processed, some old or communal edges are deleted,
and thus we remove them from $\oldd$. Notice that no edge is inserted 
to $\oldd$ by the definition of old and communal edges.

In addition to this, we store the old edges on each boundary component of the  
faces of $\oldm$ in a concatenable queue. Notice that such edges are not necessarily 
contiguous on the boundary component. In spite of this fact, 
we can traverse the old edges along a boundary component of each face of $\oldm$
in time linear in the number of the old edges due to the concatenable queu for the old edges.

\subparagraph{Fully dynamic point location for new and communal edges.}
\label{subsec:new-ds}
Let $\newce$ be the set of all new and communal edges and
$\newm$ be the subdivision induced by $\newce$.
Also, let $\oldce$ denote the set
of all old and communal edges and $\oldm$ be the subdivision induced
by $\oldce$.

We maintain a dynamic data structure that supports vertical
ray-shooting queries for $\newce$.  The update time $U(n)$ is
$O(\log n\log\log n)$ and the query time $Q(n)$ is
$O(\log n(\log\log n)^2)$ if we use the data structure by Chan and
Nekrich~\cite{cn-locatoin-2015}. Or, there are alternative data
structures with different update and query
times~\cite{ABG-Improved-2006,BJM-Dynamic-1994,Cheng-NewResults-1992}.

We also maintain the boundary of each face $F$ of $\newm$.
We store each connected component of the boundary of $F$ in a concatenable queue.
More specifically, a concatenable queue represents a cyclic sequence of 
the edges in a connected component of the boundary of $F$.
Since $e$ is incident to at most two
faces of $\newm$, there are at most two such elements in the
queues. We implement the concatenable queues using the 2-3 trees. 
We choose an element in each queue and call it the \emph{root} of the queue. For a concatenable queue implemented by a 2-3 tree,
	we choose the root of the 2-3 tree as the root of the queue.
	Given any element of a queue, we can access the root of 
	the queue in $O(\log n)$ time. 
For an inner boundary of a face $F$ of $\newm$, we let the root of the queue
for this inner boundary point to the root of the queue for the outer boundary of $F$. We also make the root of
the queue for the outer boundary of $F$ point to the root of the queue for all inner 
boundaries of $F$.
Also, we let each edge of $\newce$ point to its corresponding elements in
the queues.

We maintain a balanced binary search tree on the vertices of $\newm$ sorted in a
lexicographical order so that we can check whether a point in the plane is a
vertex of $\newm$ in $O(\log n)$ time.  Also, for each vertex of $\newm$,
we maintain a balanced binary search tree on the edges incident to it in $\newm$
in clockwise order around it. The update procedure of this data structure is straightforward,
and the update time is subsumed by the time for maintaining
the boundaries of the faces of $\newm$. Thus, in the following, we do not
mention the update of this structure.

\begin{lemma}\label{lem:size-ds-new}
	The data structures $\oldd$ and $\newd$ have size $O(n)$.
\end{lemma}

\section{Update procedures for fully dynamic point location}\label{sec:update}
We have two update operations: $\InsertEdge(e)$ and $\DeleteEdge(e)$.
Recall that we rebuild the data structures periodically. More precisely, we
reconstruct the data structures if we have processed $f(n)$ updates since the latest
reconstruction time, where
$n$ is the number of the edges we have at the moment.  After the
reconstruction, the data structure $\newd$ becomes empty.  This is
simply because the reconstruction resets all edges to old.  For the
data structure $\oldd$, we will show that the amortized time for
reconstruction is $O(n\log n/f(n))$. Also, this data structure is
updated as some old or communal edges are deleted.

In this section, we present a procedure for updates of the two data
structures. Recall that we use $\oldm$ to denote the subdivision induced by the
old and communal edges, and $\newm$ to denote the subdivision induced
by the new and communal edges. We use the subdivisions,
$\oldm$ and $\newm$, only for description purpose, and we do not
maintain them.

\subsection{Common procedure for edge insertions and edge deletions}
We are given operation $\InsertEdge(e)$ or $\DeleteEdge(e)$ for an
edge $e$.  Recall that we construct $\oldd$ and $\newd$ periodically.
The reconstruction period $f: \mathbb{N}\rightarrow \mathbb{N}$ is an
increasing function satisfying that $f(n)/2\leq f(n/2)\leq n/4$ for
every $n$ larger than a constant. 

During the process, we receive update operations.
We use an integer
index $i$ to denote the interval between the time when we receive
the $i$th operation and the time when we receive the $(i+1)$th
operation.
We consider two consecutive reconstructions that occur at time
$i$ and at time $j$.
Let $n_i$ and $n_j$ denote the numbers of edges at times $i$ and
$j$, respectively.

\begin{lemma} \label{lem:update-old} The amortized reconstruction time
	of $\oldd$ is $O(n\log n/f(n))$, where $n$ is the number of all
	edges at the moment.
\end{lemma}
\begin{proof}
	Consider the two consecutive reconstructions occur at time stamps
	$ i$ and $ j$ with $ i< j$.
	The amount of time for the reconstruction at time $ j$ is
	$O(n_j\log n_j)$.  Imagine that the reconstruction time is
	distributed equally to the update operations occurring between
	$ i+1$ to $ j$.  Then the amortized reconstruction time per update
	operation in this period is $O(n_j\log n_j / ( j- i))$.
	
	We claim that $O(n_j\log n_j / ( j- i))$ is $O(n\log n/ f(n))$,
	where $n$ is the number of edges at any fixed time in the period
	from time $ i+1$ to time $ j$.  By the reconstruction scheme, we have
	$f(n_j)= j- i$.  We also have $n-f(n_j) \leq n_j \leq n+f(n_j)$
	because each update operation inserts or deletes exactly one edge to
	the current subdivision.  Since $f(n_j) \leq n_j/2$, we have
	$n_j\leq 2n$.  Thus, we have
	$n_j\log n_j/( j- i) = n_j\log n_j / f(n_j) \leq 2n\log
	2n/f(n_j).$
	
	The only remaining thing is to show that $f(n_j) \geq c\cdot f(n)$
	for some constant $c>0$.  We first claim that $n_j\geq n/2$. Assume
	to the contrary that $n_j< n/2$.  Since $n-f(n_j) \leq n_j$ and
	$f(\cdot)$ is increasing, we have $n/2 < f(n_j)< f(n/2)$, which
	contradicts that $f$ satisfies $f(n/2)\leq n/4$. Therefore,
	$n_j\geq n/2$. Since $f(\cdot)$ is increasing, we have
	$f(n_j) \geq f(n/2) \geq f(n)/2$.  Therefore, the lemma holds.
\end{proof}

The insertion or deletion sets some old edges to communal.  By applying
a point location query for an endpoint of $e$ in $\oldm$, we find
the faces $F$ of $\oldm$ such that the boundary of $F$ contains an
endpoint of $e$ or the interior of $F$ is intersected by $e$.  All
edges lying on the outer boundary and at most one inner boundary of $F$ 
become communal.  We insert
them and $e$ to the data structure $\newd$ for vertical ray shooting
queries. This takes $O(N\cdot U(n))$ time, where $N$ denotes the
number of all edges inserted to the data structure. 
It is possible that some edges of the faces are already communal.  In
this case, we avoid removing (also accessing) 
such edges by using the concatenable queue representing the cyclic 
sequence of the old edges on each boundary component of $F$.
%

\begin{lemma}\label{lem:amortized-number}
	The average number of old edges which are set to communal
	is $O(n/f(n))$ at any moment, where $n$ is all edges at the
	moment.
\end{lemma}
\begin{proof}
	Consider the two consecutive reconstructions occur at time stamps
	$ i$ and $ j$ with $ i< j$.  
	In the period from time $ i$ to $ j$, at most $n_i$
	old edges are set to communal in total because
	each edge we have at time $ i$ is set to communal at most
	once. Imagine that the number of old edges which are set to communal is
	distributed equally to the update operations occurring
	between $ i+1$ to $ j$. 
	Then the average number of old edges which are set to
	communal per update operation in this period is $O(n_i / ( j- i))$.
	
	We claim that $n_i / ( j- i)=O(n/f(n))$,
	where $n$ is the number of edges at any fixed time in the
	period from time $ i+1$ to time $ j$.  By the
	reconstruction scheme, we have $f(n_j)= j- i$.  We also have
	$n-f(n_j) \leq n_i \leq n+f(n_j)$
	because each update operation
	inserts or deletes exactly one edge.  Since $f(n_j) \leq n_j/2$ and $n/2\leq n_j\leq 2n$ as we already
	showed in Lemma~\ref{lem:update-old}, we
	have $n_i\leq 2n$.  Thus, we have
	$n_i /( j- i) = n_i / f(n_j) \leq
	2n/ f(n_j)=O(n /f(n))$.
\end{proof}

\begin{corollary}\label{cor:amortized-insertion}
	The amortized time for inserting new and communal edges to the
	vertical ray shooting data structure in $\newd$ is
	$O(n\cdot U(n)/f(n))$.
\end{corollary}
\begin{proof}
	By Lemma~\ref{lem:amortized-number}, the average number of
	old edges which are set to communal is $O(n/f(n))$.  We use the
	notations defined in the proof of Lemma~\ref{lem:amortized-number}.
	In the period between $ i$ and $ j$, we insert $O(n_i)$ edges to
	the vertical ray shooting data structure in total because each edge
	we have at time $ i$ is inserted to the structure at most
	once. Thus the amortized update time of $\newd$ at any moment in
	this time is $O(n\cdot U(n')/f(n))$, where $n'$ is the maximum
	number of the edges we have in the period.
	
	We claim that $n'\leq 2n$.  Since $n'-f(n_j) \leq n$
	and $n_j\leq 2n$, we have
	$n' \leq f(n_j)+n \leq f(2n)+n \leq 2f(n)+n \leq 2n$.  Assuming that
	$U(n)$ is an increasing sublinear function, we have
	$U(n')=U(2n) =O(U(n))$, which implies the lemma.
\end{proof}


\subsection{Edge insertions}\label{subsec:edge-insertion}
We are given operation $\InsertEdge(e)$ for an edge $e$.  For the data
structure $\oldd$, we do nothing since the set of the old and communal
edges remains the same.  For the data structure $\newd$, we are
required to insert one new edge $e$ and several communal edges.  In
other words, we are required to update 
the ray shooting data structure, 
the concatenable queues 
and the pointers associated to each edge of
$\newce$.  We first process the update due to the communal edges, and
then process the update due to the new edge $e$.  The process for the
new edge $e$ is the same as the process for the communal edges, except
that there is only one new edge $e$, but there are $O(n/f(n))$
communal edges (amortized).  In the following, we describe the process
for the communal edges only.

Let $\mathsf{\overline{E}_n}$ be the union of the closures of all
edges of $\newce$, where $\newce$ is the set of the new and communal
edges before $\InsertEdge(e)$ is processed.  Recall that it is not
necessarily connected. 
Recall that the old edges on the outer boundary of the face intersected by $e$
become communal. If the outer boundary is connected to an inner boundary, we also set
the edges of the inner boundary to communal. In this case, let $\gamma$ be the 
cycle consisting of these two boundary components and $e$. Otherwise, 
let $\gamma$ be the outer boundary of the face
in $\oldm$ intersecting $e$. If $\gamma$ consists of only communal edges, we do
nothing. Thus we assume that it contains at least one old edge.  We
insert the old edges of $\gamma$ to $\newd$ in
$O(n\cdot U(n)/f(n))$ amortized time by
Corollary~\ref{cor:amortized-insertion}.  Recall that the average
number of such edges is $O(n/f(n))$ by
Corollary~\ref{cor:amortized-insertion}.  

Now we update the concatenable queues and the pointers made by the communal edges on $\gamma$ in
$O(f(n) Q(n)+n\log n/f(n))$.  Let $F$ be the face of $\newm$
intersected by $\gamma$.

\begin{lemma}\label{lem:face-tree-communal}
	The curve $\gamma$ intersects no connected component of
	$\mathsf{\overline{E}_n}$ enclosed by
	$\gamma$ assuming that $\gamma$ contains at least one old edge.
\end{lemma}
\begin{proof}
	Assume that $\gamma$ intersects a
	connected component of $\mathsf{\overline{E}_n}$ enclosed by
	$\gamma$. Since $\gamma$ comes from $\oldm$, neither communal nor old edge is incident to $\gamma$.  
	Thus, there is a new edge incident to $\gamma$. Consider the subdivision induced by all new
	and communal edges restricted to the region enclosed by $\gamma$ 
	after $\gamma$ is inserted. In this subdivision, there is a face $F$ incident
	to both a new edge $\newe$ and an old edge $\olde$.
	Note
	that $\newe$ is inserted after the latest reconstruction.  When it
	was inserted, all outer boundary edges of the face $F'$ that was
	intersected by $\newe$ in the subdivision $\mathsf{M'}$ of the old
	and communal edges at the moment were set to communal.  Since a
	communal edge is set to old only by a reconstruction, the only
	possibility for $\olde$ to remain as old is that $\olde$ was not on
	the outer boundary of $F'$ but on the outer boundary of another face
	in $\mathsf{M'}$, and after then it has become an outer boundary
	edge of $F$ by a series of splits and merges of the faces that are
	incident to $\olde$. These splits and merges occur only by
	insertions and deletions of edges, and $\olde$ is set to communal by
	such a change to the face that is incident to $\olde$, and remains
	communal afterwards. Therefore, $e_o$ is not old, which is a
	contradiction.
\end{proof}

By Lemma~\ref{lem:face-tree-communal}, there is a unique face
$F_\gamma$ in the subdivision $\newm$ after the communal edges are
inserted such that the outer boundary of $F_\gamma$ is $\gamma$. See
Figure~\ref{fig:communal}.  The boundaries of $F$ change due to the
communal edges, but the boundaries of the other faces remain the same.
More precisely, $F$ is subdivided into subfaces, one of which is $F_\gamma$. 
We compute the concatenable
queues for each boundary component of the subfaces. 
We first show how to do this for $F_\gamma$. Then we show
how to compute all boundaries of every subface in the same time
while computing the boundary of $F_\gamma$.

\begin{figure}
	\begin{center}
		\includegraphics[width=0.7\textwidth]{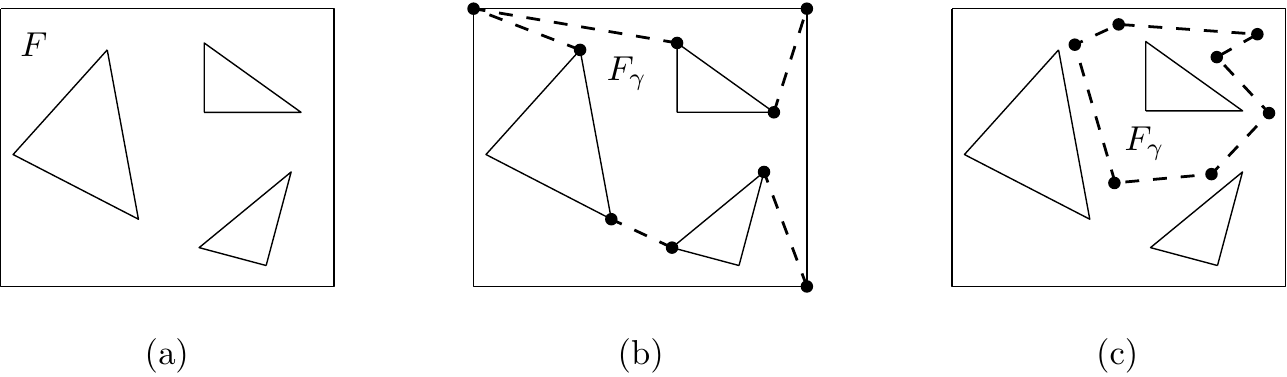}
		\caption{\small(a) Subdivision $\newm$ before the communal edges are
			inserted.  (b) The dashed edges are communal edges made by
			$\InsertEdge$.  They subdivide $F$ into three faces one of which
			is $F_\gamma$. The boundary of $F_\gamma$ is $\gamma$. (c) All edges of $\gamma$ are communal edges set
			by $\InsertEdge$. \label{fig:communal}}
	\end{center}
\end{figure}

\subparagraph{Concatenable queue for the outer boundary $\gamma$ of
	$F_\gamma$.}  We walk along the old edges of $\gamma$ which become
communal one by one using the concatenable queue for the old edges of
$\gamma$. We make an empty concatenable queue for $\gamma$, and insert
such edges one by one. If two consecutive old edges $g_1$ and $g_2$ of
$\gamma$ share no endpoints, there is a polygonal chain between $g_1$
and $g_2$ of $\gamma$ consisting of communal edges only. Notice that
this chain is a part of a boundary component of $F$. We find the boundary
component of $F$ in constant time. We split it 
with respect to $g_1$ and $g_2$, and combine
one subchain with the concatenable queue for $\gamma$.  We keep the
other subchain for updating the boundary of $F$.  In this way, we can
obtain the concatenable queue for the outer boundary $\gamma$ of
$F_\gamma$.  This takes $O(N\log n)$ time, where $N$ is the number of
old edges of $\gamma$ which become communal.

\subparagraph{Concatenable queues for the inner boundaries of
	$F_\gamma$.}  A inner boundary $\beta$ of $F$ might be enclosed by $F_\gamma$ in
the subdivision after the communal edges are inserted.  For each inner
boundary $\beta$ of $F$, we
check if it is enclosed by $F_\gamma$.  To do this, we compute
the edge $e'$ immediately lying above the topmost vertex of $\beta$ using
the vertical ray shooting data structure on all new and communal
edges, which include the edges of $\gamma$.  Using the pointer for
each edge $e'$ pointing to the elements in the concatenable
queues, we can find the boundary component $\beta_{e'}$
containing $e'$ in constant time.  If $\beta_{e'}$
is $\gamma$, we can determine if $\beta$ is enclosed by
$F_\gamma$ immediately.  Otherwise, $\beta$ is enclosed by 
$F_\gamma$ if and only if $\beta_{e'}$ is enclosed by $F_\gamma$.
Therefore, we can determine for each inner boundary $\beta$ of $F$ whether
$\beta$ is enclosed by $F_\gamma$ in $O(f(n)Q(n))$ time in total
since there are $O(f(n))$ inner boundaries of $F$ by Lemma~\ref{lem:new-face}.

Since each inner boundary of $F_\gamma$ was an inner boundary of $F$
before the communal edges are inserted, there is a
concatenable queue for each inner boundary of $F_\gamma$ whose root
node 
points to the outer boundary of $F$. We make the root of the concatenable queue to point
to $F_\gamma$, which takes $O(f(n))$ time in total for all inner
boundaries of $\gamma$.

\subparagraph {Concatenable queues for the boundary of $F$.}
We walk along the old edges of $\gamma$ one by one. Here we compute
all endpoints of the old edges of $\gamma$ that are already in $\newm$ 
using the balanced binary search tree on the vertices of $\newm$ in $O(N\log n)$ time
in total. For each old edge $e'$ with endpoint $u$ on $\newm$, 
we locate its position on the balanced binary search
tree on the edges incident to $u$ in $\newm$ in the same time.
Then we can find the boundary components of $F$ containing $u$.
For each such boundary component, we split its concatenable queue with respect
to the vertices of $\gamma$ on it.
In this way, we obtain a number of pieces of the boundary components of $F$.
By combining them with the old edges of $\gamma$, we can obtain
the outer boundaries of the subfaces of $F$ in $O(N\log n)$ time in total.

If $F$ is subdivided into exactly two subfaces, the outer boundary of one subface $F'$ 
is the outer boundary of $F$, and the outer boundary of the other is $\gamma$.
An inner boundary of $F'$ consists of old edges of $\gamma$ and parts of inner boundaries
of $F$. We can obtain the concatenable queue for this inner boundary similarly in $O(N\log n)$ time.
Then we are done in this case.

In the case that
$F$ is subdivided into more than two subfaces, 
we compute the other inner boundaries of each subface. Notice that 
they are inner boundaries of
$F$ before $\gamma$ is inserted. We find the subface enclosing each inner boundary of $F$.  
This can be done in $O(f(n) Q(n))$ time in total for
every inner boundary edge of $F$ as we did for inner
boundaries of $F_\gamma$. More precisely, for the topmost vertex of each inner boundary $\beta$ of $F$,
we find the edge $e'$ lying immediately above it among all new and communal edges including $\gamma$.
If $e'$ is an outer boundary edge of a subface, $\beta$ is an inner boundary of the subface.
Otherwise, $e'$ is an edge of an inner boundary $\beta'$ of $F$ before $\gamma$ is inserted.
Then $\beta$ is an inner boundary of a subface $F'$ if and only if
$\beta'$ is an inner boundary of $F'$.
Since there are $O(f(n))$ inner boundaries of $F$, this takes $O(f(n)Q(n))$ time in total.

\subparagraph{Pointers for edges.}  Finally, we update the pointers associated
to each edge of $\newce$ and each old edge of $\gamma$ which become
communal.  Recall that each edge of $\newce$ points to the element in
the concatenable queues corresponding to it. For the update of the
concatenable queues, we do not remove the elements of them. We just
make their pointers to point to other elements of the queues.
Therefore, we do not need to do anything for $\newce$. The only thing
we do is to make each old edge of $\gamma$ to point to the elements in
the queues, one representing the outer boundary of
$F_\gamma$ and one representing a boundary component of a subface of
$F$.  This takes $O(N\log n)$ time, where $N$ is the number of old
edges of $\gamma$ which become communal.

\medskip Therefore, the overall update time for inserting the communal
edges is $O(f(n) Q(n)+N\log n)$. Since the average value 
of $N$ is $O(n/f(n))$ by Lemma~\ref{lem:amortized-number}, the
amortized update time is $O(f(n) Q(n)+n\log n/f(n))$.
Similarly, the update time for the insertion of $e$ is
$O(f(n) Q(n) + \log n)$.
The amortized reconstruction time is
$O(n\log n/f(n))$ by Lemma~\ref{lem:update-old}.  Also, the amortized
time for inserting the communal edges to the vertical ray shooting
data structure is $O(n\cdot U(n)/f(n))$ by
Corollary~\ref{cor:amortized-insertion}.  Therefore, the overall
update time is $O(f(n) Q(n)+n\log n/f(n) + n\cdot U(n)/f(n))$,
which is $O(f(n) Q(n) + n\cdot U(n)/f(n))$ since
$U(n)=\Omega(\log n)$.

\begin{lemma}
	We can process $\InsertEdge(e)$ in
	$O(f(n) Q(n)+n\cdot U(n)/f(n))$ amortized time.
\end{lemma}

\subsection{Edge deletions}
We are given operation $\DeleteEdge(e)$ for an edge $e$.  For the data
structure $\oldd$, we update the semi-dynamic point location data
structure on the old and communal edges. We also update the
concatenable queues for old edges on the boundary components of a face
of $\oldm$.

For the data structure $\newd$, we are required to update the
concatenable queues and the pointers
associated to each edge of $\newce$. Here, we insert $O(n/f(n))$
communal edges and delete only one edge $e$ from the data structure.
The insertion of the communal edges is exactly the same as the case
for edge insertions in the previous subsection.  The deletion of $e$
is also similar to the update procedure for edge insertions.  Consider
the union $\mathsf{\overline{E}_n}$ of the closures of all edges of
$\newce\setminus\{e\}$.  There are four cases on the configuration of
$e$: both endpoints of $e$ lie on the same connected component of
$\mathsf{\overline{E}_n}$, the two endpoints of $e$ lie on two
distinct connected components of $\mathsf{\overline{E}_n}$, only one
endpoint of $e$ lies on $\mathsf{\overline{E}_n}$, or no endpoint of
$e$ lies on $\mathsf{\overline{E}_n}$. See Figure~\ref{fig:insert}.

\begin{figure}
	\begin{center}
		\includegraphics[width=0.7\textwidth]{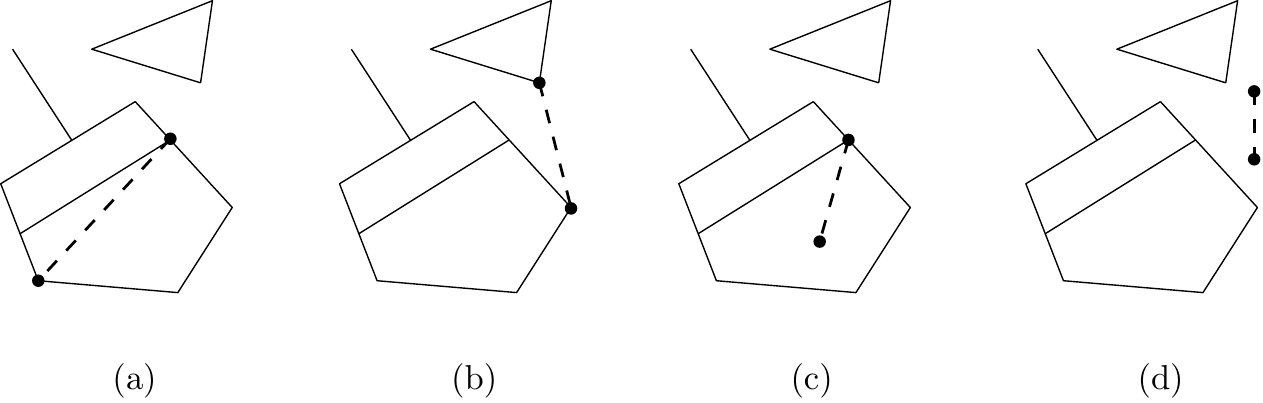}
		\caption{\small(a) The endpoints of $e$ (dashed line segment) lie
			on the same connected component of $\mathsf{\overline{E}_n}$.
			(b) The endpoints of $e$ lie on distinct connected components of
			$\mathsf{\overline{E}_n}$. (c) Only one endpoint of $e$ lies on
			$\mathsf{\overline{E}_n}$. (d) No endpoint of $e$ lies on
			$\mathsf{\overline{E}_n}$. \label{fig:insert}}
	\end{center}
\end{figure}

\subparagraph{Case 1: Two faces $F_1$ and $F_2$ are merged into one
	face $F$.}  We can find two faces $F_1$ and $F_2$ in $\newm$
incident to $e$ in constant time using the pointers that $e$ has.  We
merge the concatenable queues, one for the outer boundary of $F_1$ and
one for the outer boundary of $F_2$, after removing $e$ from the
queues in $O(\log n)$ time.  The resulting queue represents the outer boundary of $F$.
Every inner boundary of $F_1$ and $F_2$ becomes an inner
boundary of $F$.  Thus we make the concatenable
queues representing the inner boundaries of $F_1$ and $F_2$ to point
to $F$.  The other concatenable queues remain the same.
Since $F$ has $O(f(n))$ inner boundaries by Lemma~\ref{lem:new-face}, the
update of the pointers for the concatenable queues takes
$O(f(n))$ time.

Then we update the pointers associated to each edge of
$\newce\setminus\{e\}$. Recall that each edge of $\newce$
points to the elements in the
concatenable queues corresponding to it. For the update of the
concatenable queues, we remove the elements for $e$ only. For the
other elements, we just make their pointers to point to other elements
of the queues.  Therefore, we do not need to do anything for
$\newce\setminus\{e\}$.

\subparagraph{Case 2: A boundary component $\gamma$ of $F$ is
	split into two boundary components of $F$.}  We find the boundary component of $F$
containing $e$ in constant time using the pointers that $e$ has.  
It is possible that the component is an inner boundary of $F$ or an outer boundary of $F$.
For the first one, it is split into two inner boundaries of $F$.
For the second one, it is split into one inner boundary of $F$ and one outer boundary of $F$.
In any case, we delete $e$ from its concatenable queue and split the queue 
into two pieces with respect to $e$ in $O(\log n)$ time. We make each queue
to point to the concatenable queue for a boundary component of $F$ accordingly.
The other concatenable queues remain the same.

\subparagraph{Case 3: An edge of a boundary component $\gamma$
	disappears.}  We  update the concatenable queue for this boundary
in $O(\log n)$ time by deleting the element corresponding to $e$ from the concatenable queue. 
We do nothing for updating
the pointers associated to each edge of $\newce\setminus\{e\}$. In
total, Case~3 can be handled in $O(\log n)$ time.

\subparagraph{Case 4: A boundary component $\gamma$ disappears.}  
We simply remove the concatenable queue representing this boundary
component in constant time. It is an inner boundary of a face $F$ of $\newm$. 
The outer boundary of $F$ and $\gamma$ point to each other. We 
remove the pointers in constant time.

\medskip Therefore, the update time for the deletion of $e$ (excluding
the time for the insertion of the communal edges) is
$O(f(n) + \log n)$. The update time for inserting the communal edges
is $O(n\cdot U(n)/f(n)+f(n) Q(n))$.  The amortized
reconstruction time is $O(n\log n/f(n))$.  The deletion of $e$ from
the vertical ray shooting data structure takes $O(U(n))$ time.  Also,
the deletion of $e$ from $\oldd$ takes $O(\alpha(n))$ time.
Therefore, the overall update time is
$O(f(n) Q(n)+n\log n/f(n) + n\cdot U(n)/f(n))$, which is
$O(f(n) Q(n)+n\cdot U(n)/f(n))$ since $U(n)=\Omega(\log n)$.

\begin{lemma}
	We can process $\DeleteEdge(e)$ in
	$O(f(n) Q(n)+n\cdot U(n)/f(n))$ time.
\end{lemma}

\section{Query procedure}\label{sec:query}

We call the subdivision induced by all old, communal, and new
edges the \emph{complete subdivision} and denote it by $\compm$.
Sometimes we mention a face without specifying the subdivision if
the face is in the complete subdivision.

Given $\oldd$ and $\newd$, we are to answer $\locate(x)$, that is, to
find the face containing the query point $x$ in $\compm$.
Let $\oldf$ and $\newf$ be the faces of $\oldm$
and $\newm$ containing $x$, respectively.  By
Theorem~\ref{thm:decrement}, we can find $\oldf$ in $O(\log n)$ time.
For $\newf$, we find the edge $e$ of $\newm$ immediately lying above $x$
in $Q(n)$ time using the vertical ray shooting data structure. Then we
find the faces containing $e$ on their boundaries using the pointers
$e$ has.  There are at most two such faces.  Since the connected
components of the boundary of each face are oriented consistently, we
can decide which one contains $x$ in constant time. Therefore, we can
compute $\newf$ in $O(Q(n))$ time in total.  To answer $\locate(x)$, we
need the following lemmas.


\begin{lemma}\label{lem:no-oldnew}
	No face of $\compm$ contains both an old edge and new edge
	in its outer boundary.
\end{lemma}
\begin{proof}
	Assume to the contrary that both a new edge $\newe$ and an old edge
	$\olde$ lie on the outer boundary of a face $F$ of $\compm$. Note
	that $\newe$ is inserted after the latest reconstruction.  When
	$\newe$ was inserted, all outer boundary edges of the face $F'$ that
	was intersected by $\newe$ in $\oldm$ at the moment
	were set to communal.  Since a communal edge is set to old only by a
	reconstruction, the only possibility for $\olde$ to remain as old is
	that $\olde$ was not on the outer boundary of $F'$ but on the outer
	boundary of another face in $\oldm$, and after then it has
	become an outer boundary edge of $F$ by a series of splits and
	merges of the faces that are incident to $\olde$. These splits and
	merges occur only by insertions and deletions of edges, and $\olde$
	is set to communal by such a change to the face that is incident to
	$\olde$, and remains communal afterwards.
	This contradicts that $\olde$ is old, and this case never occurs.
\end{proof}
\begin{lemma}\label{lem:outer}
	For any face $F$ in $\compm$, there exists a face in $\oldm$ or in
	$\newm$ whose outer boundary coincides with the outer boundary of
	$F$.
\end{lemma}
\begin{proof}
	By Lemma~\ref{lem:no-oldnew}, the outer boundary of $F$ contains
	either no new edge or no old edge.
	Consider the case that the outer boundary of $F$ contains no new
	edge, that is, all edges on the outer boundary of $F$ are in $\oldm$.
	No edge of $\oldm$ intersects $F$ because all
	edges of $\oldm$ appear in $\compm$ but no edge of $\compm$
	intersects $F$.  Thus there is a face in $\oldm$ whose outer boundary
	coincides with the outer boundary of $F$.
	Similarly, we can prove the lemma for the
	case that the outer boundary of $F$ contains no old edge.
\end{proof}


Using the two lemmas above, we can obtain the following lemma.
\begin{lemma}\label{lem:locate}
	For any two points in the plane, they are in the same face in $\compm$
	if and only if they are in the same
	face in $\oldm$ and in the same face in $\newm$.
\end{lemma}
\begin{proof}
	Assume that two points $x$ and $y$ are in the same face $F$ in $\compm$.
	There is a face $F'$ in $\oldm$ or in $\newm$
	whose outer boundary coincides with the outer boundary of $F$ by
	Lemma~\ref{lem:outer}. Consider a face $F''$ in $\oldm$ or $\newm$ enclosed by
	$F'$.  Neither $x$
	nor $y$ is enclosed by $F''$ because the edge set of $\oldm$ (and $\newm$) is a subset
	of the edge set of $\compm$. Since the complete
	subdivision $\compm$ is planar, this implies that $x$ and $y$ are in the same
	face in both $\oldm$ and $\newm$.
	
	Now assume that two points $x$ and $y$ are in different faces in $\compm$.
	Let $F_x$ and $F_y$ be the faces containing
	$x$ and $y$ in $\compm$, respectively.  This means
	that $x$ is not enclosed by $F_y$ or $y$ is not enclosed by $F_x$.
	The outer boundaries of $F_x$ and $F_y$
	are distinct. By Lemma~\ref{lem:outer}, there are faces $F_x'$ and
	$F_y'$ in $\oldm$ or in $\newm$ whose outer boundaries
	coincide with the outer boundaries of $F_x$ and $F_y$,
	respectively.  Note that $x$ is enclosed by $F_x'$ and $y$
	is enclosed by $F_y'$. However, either $x$ is not enclosed by
	$F_y'$ or $y$ is not enclosed by $F_x'$. Therefore, $x$
	and $y$ are in different faces in $\oldm$ or $\newm$, which proves
	the lemma.
\end{proof}

Lemma~\ref{lem:locate} immediately gives an $O(Q(n))$-time query
algorithm. We represent the name of each face in $\compm$
by the pair consisting of two faces, one from $\oldm$ and
one from $\newm$, corresponding to it.  In this way,
we have the following lemma.

\begin{lemma}
	We can answer any query $\locate(x)$ in $O(Q(n))$ time 
	using $\oldd$ and $\newd$.
\end{lemma}

By setting $f(n)=\sqrt{n\cdot U(n)/Q(n)}$, we have the following
theorem.
\begin{theorem}
	We can construct a data structure of size $O(S(n))$ so that
	$\locate(x)$ can be answered in $O(Q(n))$ time
	for any point $x$ in the plane. Each update,
	$\InsertEdge(e)$ or $\DeleteEdge(e)$, can be processed in
	$O(\sqrt{n \cdot U(n) \cdot Q(n)})$ amortized time, where $n$ is the number
	of edges at the moment.
\end{theorem}

Using the data structure by Chan and Nekrich~\cite{cn-locatoin-2015},
we set $S(n)=n$, $Q(n)=\log n(\log\log n)^2$ and
$U(n)=\log n\log\log n$.

\begin{corollary}
	We can construct a data structure of size $O(n)$ so that
	$\locate(x)$ can be answered in $O(\log n(\log\log n)^2)$
	time for any point $x$ 
	in the plane. Each update, $\InsertEdge(e)$ or
	$\DeleteEdge(e)$, can be processed in $O(\sqrt{n}\log n(\log\log n)^{3/2})$
	amortized time, where $n$ is the number of edges at the moment.
\end{corollary}

\end{document}